\documentclass{llncs}

\usepackage[english]{babel}
\usepackage{amsmath}
\usepackage{amssymb}
\usepackage{hyperref}
\usepackage{tikz}
\usetikzlibrary{arrows,automata}
\usepackage{delarray}
\usepackage{enumerate}

\let\emptyset\varnothing

\let\implies\Rightarrow

\let\sim\prec
\newcommand{\B}{\mathbb{B}}
\newcommand{\ress}[1]{\bigl\lvert_{#1}}
\newcommand{\dom}{\mathrm{dom}}
\newcommand{\bigWiring}{\bigcup}
\newenvironment{sproof}{%
  \proof}{\endproof}

\begin{document}
\title{A framework for (de)composing with Boolean automata networks}

\author{
        K{\'e}vin Perrot\inst{1} \and
        Pac{\^o}me Perrotin\thanks{Corresponding author: 
        	\url{pacome.perrotin@lis-lab.fr}.}\inst{1} \and
        Sylvain Sen{\'e}\inst{1}
}

\date{}

\institute{
        Aix-Marseille Univ., Toulon Univ., CNRS, LIS, Marseille, France
}

\maketitle

\begin{abstract}
	Boolean automata networks (BANs) are a generalisation of Boolean cellular 
	automata. In such, any theorem describing the way BANs compute information is a 
	strong tool that can be applied to a wide range of models of computation. In this 
	paper we explore a way of working with BANs which involves adding external inputs 
	to the base model (via modules), and more importantly, a way to link
	networks together using the above mentioned inputs (via wirings). Our aim is to
	develop a powerful formalism for BAN (de)composition. We formulate two results: 
	the first one shows that our modules/wirings definition is complete; the
	second one uses modules/wirings to prove simulation results amongst BANs.
\end{abstract}



\section{Introduction}

Boolean automata networks (BANs) can be seen as a generalisation of cellular 
automata that enables the creation of systems composed of Boolean functions over any 
graph, while cellular automata only operate over lattices of any dimension. The study 
of the dynamics of a BAN, that describes the set of all computations possible in such 
a system, is a wide and complex subject. From very simple networks computing simple 
Boolean functions to possibly infinite networks able to simulate any Turing machine, 
the number of configurations always grows exponentially with the size of the network, 
making any exhaustive examination of its dynamics impractical. The study of such 
dynamics is nevertheless an important topic which can impact other fields. BANs are 
for example used in the study of the dynamics of gene regulatory 
networks~\cite{J-Demongeot2010,J-Kauffman1969,J-Thomas1973} in biology.

Many efforts to characterise the dynamics of BANs have already been put forward. For 
example, some studies~\cite{C-Alcolei2016,T-Noual2012} examine the behaviour of 
networks composed of interconnected cycles. The modularity of BANs has been studied 
from multiple perspectives. In particular from a static point of 
view~\cite{J-Alon2003,J-Milo2002}, and a functional 
one~\cite{J-Bernot2009,C-Delaplace2012,C-Siebert2009}. In this paper, we explore a 
compositional approach to BANs that allows to decompose a BAN into subnetworks called 
modules, and to compose modules together in order to form larger networks. We define 
a module as a BAN on which we add external inputs. These inputs are used to 
manipulate the result of the network computation by adding extra information. They 
can also be used to interconnect multiple modules, making more complex networks. 
Those constructions resemble the circuits described in Feder's 
thesis~\cite{T-Feder1990}, and modules can be seen as a generalisation of circuits 
over any update mode.

Section~\ref{s:motivations} discusses the possible motivations for a
(de)compositional study of BANs.
Section~\ref{s:def} introduces BANs and update modes, and Sections~\ref{s:modules} 
and~\ref{s:wirings} develop a formalism for the modular study of BANs, justified by a 
first theorem showing that any network can be created with modules and wirings. We 
also present an application of our definitions to BAN simulation in Section 
\ref{s:simulation}, leading to a second theorem stating that composing
with local 
simulations is sufficient to (globally) simulate a BAN. Finally, Section~\ref{s:examples}
presents and analyses two illustrations of the principles presented in
Section~\ref{s:motivations}.

The demonstrations of all 
results are given in appendix.

\section{Motivations}
\label{s:motivations}


BANs, despite being very simply defined locally, become complex to analyse as
the representation of their dynamics grows exponentially in the size of their
networks. BANs have been proven to be Turing-complete~\cite{M-Cook2004} and as
most of Turing-complete systems are able to show complex and emergent properties.

Yet, an important number of networks can be partially understood when viewed
through the lens of functionality (what an object is meant to achieve).
Functionality enables to use
abstraction to reduce the considered network (or some part of it) to the
computation of a function or the simulation of a dynamical system.
Assuming a functionality of the parts of a network can let us
conclude on the functionality of the network itself, at the cost of
letting aside an absolute characterisation of its dynamics (which is often
practically impossible). Such a functional interpretation aims at offering the possibility to make
verifiable predictions in a short amount of time.

It is not known if every Boolean automata network can be cut into a
reasonable amount of parts to which one can easily affect a functionality.
We will justify our present argument by illustrating it in
Section~\ref{s:examples}.

\section{Boolean automata networks} 
\label{s:def}

\subsection{Preliminary notations}

Let us first describe some of the notations used throughout 
the paper. Let $f: A \to B$ be a mapping from set $A$ to set $B$. For $S \subseteq A$ 
we denote $f(S) = \{ b \in B \mid \exists a \in S, f(a) = b\}$. We denote $f 
\ress{S}$ the restriction of $f$ to the domain $S$, $f\ress{S}: S \to B$ such that 
$f\ress{S}(a) = f(a)$ for all $a \in S$. Let $\dom(f)$ be the domain of $f$, and $g \circ f$ the 
composition of $f$ then $g$. For $f$ and $g$ two functions with disjoint domains of 
definition, we define $f \sqcup g$ as the function defined such that :
\begin{equation*}
f \sqcup g(x)=\begin{array}\{{ll}.
  f(x) & \text{ if } x \in \dom(f)\\[.5em]
  g(x) & \text{ if } x \in \dom(h)
\end{array}\text{.}
\end{equation*}
We denote $\B = \{0, 1\}$ the set of Booleans. For $K$ a sequence of $m$ elements, 
the sub-sequence from the $i$-th element to the $j$-th element is denoted 
$K_{[i, j]}$. We sometimes define functions without naming them with the notation
$a \mapsto b$, signifying that for any input $a$ the function will return $b$.
For example, the function $n \mapsto 2 \times n$ is a function that takes a
number $n$ and returns the value of $n$ multiplied by $2$.

\subsection{Definitions}

A BAN is based upon a set of automata. Each automaton is defined as a Boolean 
function, with arity the size of the network. Each variable of the function of each 
automaton is meant to correspond to an automaton in the network. By considering a 
configuration of Boolean values over this network, we can compute the Boolean 
function of each automaton and obtain a Boolean value for each automaton (\emph{i.e.} 
a local state). These values can be used to update the global state of the network, 
that we call a configuration. If we decide to update the value of each automaton at 
once, the update mode is parallel. However, if only one automaton is updated at each 
time step, 
the update mode is sequential~\cite{J-Fogelman1983,L-Robert1986}.

\begin{definition}
	A \emph{configuration} on a set $S$ is a function $x: S \to \B$.
\end{definition}

A BAN $F$ defined over the set $S$ associates a Boolean function to each element
of $S$. Each of theses functions is defined from the set of all configurations
of the BAN, $S \to \B$, to the Boolean set, $\B$.

\begin{definition}
	For $S$ a set, a \emph{Boolean automata network} (BAN) $F$ is a function
$F : S \to ( S \to \B ) \to \B$.
\end{definition}

For each $s \in S$, we denote $f_s = F(s)$ the local function of automaton $s$.

For $s \in S$ we denote $x_s = x(s)$. A function $x$ is a configuration at a given 
time over the network. Thus, we can define our function $f_s$ to be part of the set 
$(S \to \B) \to \B$. This way, a BAN $F$ can be defined as a function from the set 
$S$ to the set $(S \to \B) \to \B$. We find again that the set of all BANs
over $S$ can 
simply be defined
as $S \to (S \to \B) \to \B $. For any BAN $F$ and configuration 
$x$, we can define the configuration which is computed by $F$ from $x$. A naive way 
to do so would be to define $x' = F(x)$ such that $x'_s = f_s(x)$ for every $s$; this 
definition however is very limiting: it only allows parallel updates of our system. 
In a general definition of BANs, a computation of a BAN should allow updates of only 
a subset of the functions of the network. Slight changes to the update mode of a BAN
can deeply change its computational capabilities~\cite{J-Aracena2013,J-Goles2008}. 
Most results that assume a parallel update mode cannot be applied to a sequential 
network; the reciprocal is also true. We set the following definition of an update 
over our BAN to be as general as possible.

\begin{definition}
	Any $\delta \subseteq S$ is an \emph{update} over $S$.
\end{definition}

One can apply multiple consecutive updates to a BAN to effectively execute the BAN 
over an update mode. An \emph{update mode} is simply a sequence of updates that is 
denoted $\Delta$, where $\Delta_k$ is the $k^\text{th}$ update of the sequence. 
We define the union operator between updates modes as it will be useful for the proof of our last theorem.

\begin{definition}
	Let $\Delta$, $\Delta'$ be two update modes over a set $S$. The \emph{union} of 
	$\Delta$ and $\Delta'$ denoted $\Delta \cup \Delta'$ is the update mode defined 
	as $(\Delta \cup \Delta')_k = \Delta_k \cup \Delta'_k$. The size of $\Delta \cup 
	\Delta'$ is the maximum among the sizes of $\Delta$ and $\Delta'$.
\end{definition}

We assume that $\Delta_k = \emptyset$ if $k$ is greater than the size of $\Delta$. 
Given an update $\delta$, we can define the endomorphism $F_\delta$ over the set of 
all configurations. For every configuration $x$, we set $F_\delta(x)(s) = f_s(x)$ if 
$s \in \delta$, and $F_\delta(x)(s) = x(s)$ if $s \notin \delta$. In other words, the 
value of $s$ in the new configuration is set to $f_s(x)$ only if $s \in \delta$, 
otherwise the Boolean affectation of $s$ remains $x_s$. Now, we can define the 
execution of $F$ in a recursive way.

\begin{definition}
	The \emph{execution} of $F$ over $x$, under the update mode $\Delta$, is 
the function $F_\Delta : (S \to \B) \to (S \to \B)$ defined 
	as $F_{\Delta[1, k]}(x) = F_{\Delta_k} ( F_{\Delta[1, k-1]}(x))$, with 
	$F_{\Delta[1, 1]}(x) = F_{\Delta_1}(x)$.
\end{definition}

Throughout this paper we represent BANs as graphs called interaction graphs.
Interaction graphs are a classical tool in the study of BANs. For a BAN
$F$ defined over $S$, the interaction graph of $F$ is the oriented graph
$G = (S, \epsilon)$, where $(s, s') \in \epsilon$ if and only if the
variable $x_s$ influences the computation of the function $F(s')$.

\section{Modules} 
\label{s:modules}

Modules are BANs with external inputs. Such inputs can be added to any local function 
of a module, and any local function of a module can have multiple inputs. When a 
local function has $n$ inputs, the arity of this function is increased by $n$. These 
new parameters are referred to by elements in a new set $E$: the elements of $E$ 
describe the inputs of the module; those of $S$ describe the internal elements of the 
module. To declare which input $e \in E$ is affected to each function $f_s$, we use 
function $\alpha$.

\begin{definition}
	Let $S$ and $E$ be two disjoint sets. An \emph{input declaration} over $S$ and 
	$E$ is a function $\alpha: S \rightarrow \mathcal{P}(E)$ such that 
	$\{ \alpha(s) \mid s \in S \}$ is a partition of $E$.
\end{definition}

For each $s$, $\alpha(s)$ is the set of all external inputs of function $f_s$. 
The partition property is important because without it, some
input could be assigned to multiple nodes, or to no node at all, which
is contrary to our vision of input. To simplify notations, we sometimes denote 
$E_s = \alpha(s)$. Now, let us explicit the concept of a module.

\begin{definition}
	A \emph{module} $M$ over $(S, E, \alpha)$
is defined such that, for each $s \in S$, $M(s)$ is a function $M(s):
(S \cup E_s) \to \B$. 
\end{definition}

If $M$ is a module defined over $(S, \emptyset, s \mapsto \emptyset)$, $M$ is also a 
BAN. To compute anything over this new system, we need a configuration $x: S \to 
\B$ and a configuration over the elements of $E$.

\begin{definition}
	An \emph{input configuration} over $E$ is a function $i: E \to \B$.
\end{definition}

Let $x$ be a configuration over $S$, and $i$ an input configuration over $E$. As
$x$ and $i$ are defined over disjoint sets, we define $x \sqcup i$ as their union. 
Such an union, coupled with an update over $S$, is enough information to perform a 
computation over this new model.

\begin{definition}
	Let $x$ be a configuration over $S$ and $i$ an input over $E$. Let $\delta$ be an 
	update over $S$. The \emph{computation} of $M$ over $x$, $i$ and $\delta$, 
	denoted $M_\delta(x \sqcup i)$, is the configuration over $S$ such that 
	$M_\delta(x \sqcup i)(s) = f_s(x \sqcup i\ress{E_s})$ for each $s \in \delta$, 
	and $M_\delta(x \sqcup i)(s) = x(s)$ for every $s \in S \setminus \delta$.
\end{definition}

In the following example, we assume a total order over $S \cup E$, allowing us to 
intuitively write configurations as binary words. For example, $x = 101$ means $x(a) = 1$, $x(b) = 0$ and $x(c) = 1$.

\begin{figure}[t!] 
	\begin{center}
		\begin{tikzpicture}[->,>=stealth',shorten
>=1pt,auto,node distance=2cm, semithick, initial text=,inner sep=0pt, minimum
size=0pt] 
			\node[state] (A) {a}; 
			\node[state] (B) [right of=A] {b};
			\node[state] (C) [right of=B] {c};

			\path (B) edge (A) (C) edge (B); 
			\path (B) edge [loop below] (B); 
			\path (C) edge [loop above] (C);

			\draw (-0.5, 0.5) -- (A); 
			\draw (0, 0.7) -- (A); 
			\draw (0.5, 0.5) -- (A); 
			\node at (-0.6, 0.6) {$a_1$}; 
			\node at (0, 0.8) {$a_2$}; 
			\node at (0.6, 0.6) {$a_3$};

			\draw (1.7, 0.65) -- (B); 
			\draw (2.3, 0.65) -- (B); 
			\node at (1.7, 0.8) {$b_1$}; 
			\node at (2.33, 0.8) {$b_2$};

			\draw (4.7, 0) -- (C); 
			\node at (4.87, -0.02) {$c_1$};
		\end{tikzpicture}
	\end{center} 
	\caption{Interaction graph of the module detailed in Example~\ref{ex:1}.}
	\label{figMod}
\end{figure}
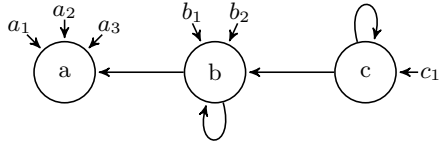

\begin{example}\label{ex:1}
	$S = \{a, b, c\}$, and $E = \{a_1, a_2, a_3, b_1, b_2, c_1\}$. We define $\alpha$ 
	such that $\alpha(a) = \{a_1, a_2, a_3\}$, $\alpha(b) = \{b_1, b_2\}$ and 
	$\alpha(c) = \{c_1\}$. Let $M$ be a module over $(S,E,\alpha)$, such that $M(a) =
	x_b \vee a_1 \vee a_2 \vee a_3$, $M(b) = \neg x_b \vee x_c \vee \neg b_1 \wedge
	b_2$, and $M(c) = \neg c_1$. Let $x = 101$, $i = 000010$ and $\delta = \{a, b\}$.
	We get that $M_\delta(x \sqcup i) = M_{\{a, b\}}(101 \sqcup 000010)$ is such that 
	$M_\delta(x \sqcup i)(a) = f_a(x \sqcup i\ress{E_a}) = 0$, $M_\delta(x \sqcup i)
	(b) = f_b(x \sqcup i\ress{E_b}) = 1$, and $M_\delta(x \sqcup i)(c) = x(c) = 1$. 
	Therefore $M_\delta(x \sqcup i) = 011$. A representation of this module is 
	pictured in Figure~\ref{ex:1}.
\end{example}

Let us now define executions, while considering that the input configuration can 
change over time.

\begin{definition}
	Let $t > 1$. Let $I = (i_1, i_2, \ldots, i_{t - 1})$ be a sequence of input 
	configurations over $E$, $X = (x_1, x_2, \ldots, x_t)$ a sequence of
	configurations over $S$, and $\Delta$ an update mode over $S$ of size $t$. 
	$(X, I, \Delta)$ is an \emph{execution} of $M$ if for all $1 \leq k < t$, $x_{k + 
	1} = M_{\Delta_k}(x_k \cup i_k)$.
\end{definition}

This definition allows for variation over the inputs over time. As this 
particular feature is not needed throughout this paper, we also propose a
simpler definition of executions over modules which only allows
fixed input values over time.

\begin{definition}
	Let $i$ be an input configuration over $E$. The \emph{execution} of $M$ over $x 
	\cup i$ with update mode $\Delta$ is an endomorphism over the set of all 
	configurations, denoted $M_\Delta$. It is defined as $M_{\Delta[1, k]}(x \sqcup 
	i) = M_{\Delta_k} ( M_{\Delta[1, k - 1]}(x \sqcup i) \sqcup i)$, with 
	$M_{\Delta[1, 1]}(x \sqcup i) = M_{\Delta_1}(x \sqcup i)$.
\end{definition}

\section{Wirings} \label{s:wirings}

The external inputs of a module can be used to encode any information. For instance, 
we could encode any periodic (or non-periodic) sequence of Boolean words into the 
inputs of a given module. We could also encode the output of a given BAN or module, 
combining in some way the computational power of both networks. Such a composition of 
modules is captured by our definition of wirings. A wiring is an operation that links
together different inputs and automata from one more or modules, thus forming 
bigger and more complex modules.

We decompose this compositional process into two different families of operators: the 
non-recursive and the recursive wirings. The first ones connect the automata of one 
module to the inputs of another; the second ones connect the automata of a module to 
its own inputs. 
A wiring, recursive or not, is defined by a partial map $\omega$ linking some inputs 
to automata. Let us first define non-recursive wirings.
\begin{definition}
	Let $M$, $M'$ be modules defined over $(S, E, \alpha)$ and $(S', E', \alpha')$ 
	respectively, such that $S, S'$ and $E, E'$ are two by two disjoint. A 
	\emph{non-recursive wiring} from $M$ to $M'$ is a partial map $\omega$ 
	from $E'$ to $S$.
\end{definition}
The new module result of the non-recursive wiring $\omega$ is denoted $M 
\rightarrowtail_\omega M'$ and is defined over $(S \cup S', E \cup E' \setminus 
\dom(\omega), \alpha_\omega)$. The input declaration of $M \rightarrowtail_\omega M'$ is
$\alpha_\omega(s) = \alpha(s) \setminus \dom(\omega)$ (in particular, 
$\alpha_\omega(s)=\alpha(s)$ if $s \in S$). Given $s \in S \cup S'$, the local 
function $M \rightarrowtail_\omega M'(s)$, denoted $f^\omega_s$, is defined as
\begin{equation*}
	f^\omega_s(x \sqcup i)=\begin{array}\{{ll}.
 		f_s(x\ress{S} \sqcup i\ress{E_s}) & \text{ if } s \in S\\[.5em]
 		f'_s(x\ress{S'} \sqcup i\ress{E'_s \setminus \dom(\omega)} \sqcup 
 		(x \circ \omega\ress{E'_s})) & \text{ if } s \in S'
	\end{array}\text{.}
\end{equation*}
In this new module, some inputs of $M'$ have been assigned to the values of some
elements of $M$. Such assignments are defined in the wiring $\omega$. For any $s \in S 
\cup S'$, the function $M \rightarrowtail_\omega M'(s)$ (denoted $f^\omega_s$) is 
defined over $(S \cup S' \cup \alpha_\omega(s)) \to \B$. In the case $s \in S'$, the 
image of $x \sqcup i$ is given by $f'_s$ which expects a configuration on 
$S' \cup E'_s$: the configuration on $S'$ is provided by $x$, and the configuration 
on $E'$ is partly provided by $i$ (on $E'_s \setminus \dom(\omega)$), and partly 
provided by $(x \circ \omega)$ (on $\dom(\omega) \cap E'_s$).
\begin{definition}
	Let $M$ be a module over $(S, E)$. A \emph{recursive wiring} of $M$ is a partial 
	map $\omega$ from $E$ to $S$.
\end{definition}
With $\omega$ defining now a recursive wiring over a module $M$, the result is similar 
if not simpler than in the definition of non-recursive wirings. The new module 
obtained from a recursive wiring $\omega$ on $M$ is denoted $\circlearrowright_\omega 
M$ and is defined over $(S, E \setminus \dom(\omega), \alpha_\omega)$ with the input
declaration defined as, for any $s \in S$, $\alpha_\omega(s) = \alpha(s) \setminus 
\dom(\omega)$. Given $s \in S$, $x$ and $i$, the local function 
$\circlearrowright_\omega M(s)$ is denoted $f^\omega_s$ and is evaluated to
$f^\omega_s(x \sqcup i) = f_s(x \sqcup i\ress{E_s \setminus \dom(\omega)} \sqcup (x 
\circ \omega\ress{E_s}))$.

Recursive and non-recursive wirings can be seen as unary and binary operators 
respectively, over the set of all modules. For any $\omega$, we can define the 
operators $\rightarrowtail_\omega$ and $\circlearrowright_\omega$. For simplicity we 
define that $M \rightarrowtail_\omega M' = \emptyset$ and $\circlearrowright_\omega M = 
\emptyset$ if the wiring $\omega$ is not defined over the same sets as $M$ or $M'$. 
Notice that both the recursive and non-recursive wirings defined by $\omega = 
\emptyset$ are well defined wiring. They define two operators, 
$\circlearrowright_\emptyset$ and $\rightarrowtail_\emptyset$, that will be useful 
later on.
\begin{property}
	The following statements hold.
	\begin{enumerate}[\em(i)\quad]
	\item $\forall M,\quad \circlearrowright_\emptyset M = M$.
	\item $\forall M, M',\quad M \rightarrowtail_\emptyset M' = M' 
		\rightarrowtail_\emptyset M$.
	\item $\forall M, M', M'',\quad M \rightarrowtail_\emptyset ( M' 
		\rightarrowtail_\emptyset M'' ) = ( M \rightarrowtail_\emptyset M' ) 
		\rightarrowtail_\emptyset M''$.
\end{enumerate}
\end{property}
For simplicity of notations, we will denote the empty non-recursive wiring as the 
union operator over modules: $M \cup M' = M \rightarrowtail_\emptyset M'$.

It is quite natural to want to put two modules together, by linking the input 
of the first to states of the second, and conversely. Our formalism allows
this operation in two steps : first, use a non-recursive wiring to connect
all of the desired inputs of the first module to states of the second
module. Then, use a recursive wiring to connect back all of the desired inputs 
of the second module to states of the first module.

We now express that recursive and non-recursive wirings are expressive enough to 
construct any BAN or module, in Theorem \ref{theorem:complete}. Our aim is to show 
that for any division of a module into smaller parts (partitioning), there is
a way to get back to the initial module using only recursive and non-recursive 
wirings.
\begin{definition}
	Let $(S, E, \alpha)$. Let $P$ be a set such that $\{S_p \mid p \in P\}$ is a partition of $S$. We define the \emph{corresponding 
	partition} of $E$ as $\{E_p = \bigcup_{s \in S_p} \alpha(s) \mid p \in P \}$.
\end{definition}
\begin{definition}
We can now develop the corresponding partition of the input declaration, and define the partition of $M$ itself.
	For every $p \in P$, we define $\alpha_p=\alpha\ress{S_p}$ over $S_p$ and $E_p$.
\end{definition}
\begin{definition}
	For every $p \in P$, let $Q_p$ verify $Q_p \cap S = \emptyset$ and $|Q_p| = |S|$, 
	and let $\tau_p: S \to Q_p$ be a bijection. For any $p \in P$, the 
	\emph{sub-module} $M_p$ over $(S_p, E_p \cup \tau_p(S \setminus S_p), \alpha_p)$ 
	is defined for $s \in S_p$ as, for all $x : S \to \B$ and for all $i: E \to \B$,
	\begin{equation*}
		M_p(s)(x\ress{S_p} \sqcup i_p) = M(s)(x \sqcup i)\text{,}
	\end{equation*} 
	where $i_p(e) = i(e)$ if $e \in E_p$ and $i_p(e) = x(\tau^{-1}_p(e))$ if $e \in 
	\tau_p(S \setminus S_p)$.
\end{definition}

In the previous definition, the purpose of each $Q_p$ is to work as a representation 
of the set $S$ for every sub-module $M_p$. Without it, every module $M_p$ would have 
used the set $(S \setminus S_p) \cup E_p$ as input set. However our definition of 
wiring requires the input sets of the wired modules to be disjoint from each other. 
The sets $Q_p$ are a workaround to bypass this technical point.

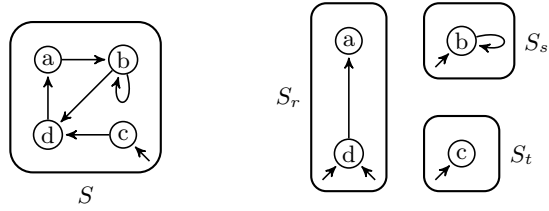
\begin{figure}[t!]
	\begin{center}
		\begin{tikzpicture}[->,>=stealth',shorten
>=1pt,auto,node distance=1cm, semithick, initial text=,inner sep=0pt, minimum
size=0pt]
			\node[state,inner sep=1pt,minimum size=10pt] (A) {a};
			\node[state,inner sep=1pt,minimum size=10pt]	(B) [right of=A] {b};
			\node[state,inner sep=1pt,minimum size=10pt]	(C) [below of=B] {c};
			\node[state,inner sep=1pt,minimum size=10pt]	(D) [below of=A] {d};

			\node[state,inner sep=1pt,minimum size=10pt,node distance = 1.5cm] at 
				(4, 0.25) (PA) {a};
			\node[state,inner sep=1pt,minimum size=10pt,node distance = 1.5cm] (PB) 
				[right of=PA] {b};
			\node[state,inner sep=1pt,minimum size=10pt,node distance = 1.5cm] (PC) 	
				[below of=PB] {c};
			\node[state,inner sep=1pt,minimum size=10pt,node distance = 1.5cm] (PD) 
				[below of=PA] {d};

			\path (A) edge (B);
			\path (B) edge [loop below] (B) edge (D);
			\path (C) edge (D);
			\path (D) edge (A);
			\draw (1.35, -1.35) -> (C);
			\draw [thick, rounded corners=.3cm] (-0.5,0.5) rectangle (1.5,-1.5);
			\node at (0.5, -1.8) {$S$};

			\path (PD) edge (PA);
			\draw (3.65, -1.6) -> (PD);
			\draw (4.35, -1.6) -> (PD);
			\draw [thick, rounded corners=.2cm] (3.5, 0.75) rectangle (4.5, -1.75);
			\node at (3.2, -0.5) {$S_r$};

			\path (PB) edge [loop right] (PB);
			\draw (5.15, -0.1) -> (PB);
			\draw [thick, rounded corners=.2cm] (5, 0.75) rectangle (6.2, -0.25);
			\node at (6.5, 0.22) {$S_s$};

			\draw (5.15, -1.6) -> (PC);
			\draw [thick, rounded corners=.2cm] (5, -0.75) rectangle (6, -1.75);
			\node at (6.3, -1.27) {$S_t$};
		\end{tikzpicture}
	\end{center}
	\caption{Interaction graphs related to Example~\ref{ex:wir}. The interaction 
		graph of the original module is on the left and the interaction graphs of the 
		partition of $M$ are on the right. Notice that we did not represent the input 
		sets $E$, $Q_r$, $Q_s$ and $Q_t$.}
	\label{figWir}
\end{figure}

\begin{example} 
\label{ex:wir}
	Let $S = \{a, b, c, d\}$, $E = \{e\}$, $P = \{r, s, t\}$ and $S_r = \{a, d\}$, 
	$S_s = \{b\}$ and $S_t = \{c\}$. For each $p \in P$, we define 
	$Q_p = \{a_p, b_p, c_p, d_p\}$. In the module $M_r$, $\alpha_r(a) = \emptyset$ 
	and $\alpha_r(d) = \{b_r, c_r\}$. In the module $M_s$, $\alpha_s(b) = \{a_s\}$. 
	In the module $M_t$, $\alpha_t(c) = \{e\}$. The modules $M_r, M_s$ and $M_t$ are 
	defined over disjoint sets and can be wired (see Figure~\ref{figWir} for an 
	illustration).
\end{example}

As a reminder, the union operator over modules is defined to be the result of
an empty non-recursive wiring.

\begin{theorem}
\label{theorem:complete}
	Let $M$ be a module and $\{M_p \mid p \in P\}$ any partition of that module,
	then there exists a recursive wiring $\omega$ such that $M =\ 
	\circlearrowright_{\omega} \left( \bigWiring_{p \in P} M_p \right)$.
\end{theorem}

\begin{sproof}
	We construct $\omega$ to wire every link lost in partition $P$.
\end{sproof}

Theorem~\ref{theorem:complete} allows to say that our definition of wiring is 
complete: any BAN or module can be assembled with wirings.
It can be reworked more 
algebraically. Let $\mathcal{M}$ denote the set of all modules (which includes 
$\emptyset$), and for any $n \in \mathbb{N}$, let $\mathcal{M}_n$ denote the set of 
all modules of size $n$ (we have $\mathcal{M} = \bigcup_{n \in \mathbb{N}} 
\mathcal{M}_n$). For any subset $A \subseteq \mathcal{M}$ we denote 
$\overline{A}^\omega$ the closure of $A$ by the set of wiring operators $\bigcup_\omega 
\{ \rightarrowtail_\omega, \circlearrowright_\omega \}$. The following result is a 
direct corollary of Theorem~\ref{theorem:complete}.
\begin{corollary} 
	The set of all modules is equal to the closure by any wiring of the set of 
	modules of size $1$,
	\begin{equation*}
		\mathcal{M} = \overline{\mathcal{M}_1}^\omega.
	\end{equation*}
\end{corollary}
Every module in $\mathcal{M}_1$ is of size $1$, but as the set of inputs $E$
of a module is not bounded, the set $\mathcal{M}_1$ is infinite. In our opinion, this corollary is enough to demonstrate that our definition of
modules and wirings is sound.

\section{Simulation} \label{s:simulation}

BANs are by nature complex systems and sometimes, we like to understand the
computational power of a subset of them by demonstrating that they are able to
simulate (or be simulated by) another subset of BANs. By \emph{simulation}, we
generally mean that a BAN is able to reproduce, according to some encoding, all the
possible computations of another BAN.

Simulation is a powerful way to understand the limitations and possibilities of
BANs. It is still difficult to prove if any two BANs simulate each other. In the
present paper our aim is to prove that the property of simulating any BAN can be
reduced in some cases to the property of locally simulating any Boolean function. 
Locally simulating a function means that a module reproduces any computation of that 
function, when the parameters of the function are encoded in the module
inputs. Our claim is that if we can locally simulate every function of a BAN, in a 
way such that the simulating modules are able to communicate with each other, then we 
can simulate the same BAN with a bigger module which is obtained by a wiring over the 
locally simulating modules. In this context, modules become a strong tool
to reduce the complexity of simulation (which is a global phenomena) to a local
scale, which is more tractable.

Let us go into further details. For $F$ a BAN over the set $S$, our aim is to
simulate $F$. For this purpose, for each $a \in S$, we create $M_a$, a module which 
is defined over some sets $(T_a, E_a, \alpha_a)$ and locally simulates the function
$f_a$. To assert this local simulation we need to define a Boolean encoding
$\phi_a$ over the configurations of $M_a$. We also need to define how these
modules communicate with each other, and in the end how they will be wired
together. For any couple $a, b \in S$ such that $a \neq b$, we
define the set $U_{a, b}$ as a subset of $T_a$. This set represents all the automata
of $M_a$ that are planned to be connected to inputs of $M_b$. We can say that
the elements of $U_{a, b}$ are the only way for the module $M_a$ to send information 
to the module $M_b$. We define which information is sent from $M_a$ to $M_b$ at any 
time with a Boolean encoding $\phi_{a, b}$ over the set of configurations on $U_{a, b}$. By definition we 
always have that if $U_{a, b} \neq \emptyset$, then $\phi_a(x\ress{T_a}) \neq \bullet \implies \phi_{a, b}(x\ress{U_{a, b}}) = \phi_a(x\ress{T_a})$. This means that if 
a module encodes an information ($\bullet$ being the absence of information, {\em i.e.} in this case $\phi_a(x\ress{T_a})$ equals $0$ or $1$), the same information is sent from that module
to each module that is meant to receive information from it.
In other words, all encodings are coherent.

Now that our modules are set to communicate with each other, we only need to wire 
them to each other. The precise nature of this wiring is defined, for every
pair $a, b \in S$ such that $a \neq b$, by the function $I_{a, b} : E_b \to U_{a, b}$ 
which we call interface between $a$ and $b$. By definition:
\begin{itemize}
\item for every $s \in U_{a, b}$, there exists $e \in E_b$ such that $I_{a, b}(e) = 
	s$ (surjectivity);
\item for every $b \in S$, $\bigsqcup_{a} I_{a, b}$ is a total map from
	$E_b$ to $\bigcup_a U_{a, b}$.
\end{itemize}
With such an interface defined for every pair $(a, b)$, the final wiring
connecting all modules together is decomposed in two steps. The first one empty-wires 
every module together, the second one applies a recursive wiring which is defined as 
the union of every interface $I_{a, b}$. The last condition that we have stated over 
the definition of an interface lets us know that the obtained module has no remaining 
inputs; it can be considered as a BAN, defined over $T = \bigcup_{a \in S} T_a$. All 
these sets are illustrated in Figure~\ref{figSim}.

\begin{figure}[t!]
	\begin{center}
		\begin{tikzpicture}[->,>=stealth',shorten
>=1pt,auto,node distance=1cm, semithick, initial text=,inner sep=0pt, minimum
size=0pt]
			\node[state,inner sep=1pt,minimum size=10pt] (A) {a};
			\node[state,inner sep=1pt,minimum size=10pt] (B) [right of=A] {b};
			\node[state,inner sep=1pt,minimum size=10pt] (C) [below of=B] {c};
			\node[state,inner sep=1pt,minimum size=10pt] (D) [below of=A] {d};

			\node[state,inner sep=1pt,minimum size=10pt] at (4, 1.5) (TAa) {e};
			\node[state,inner sep=1pt,minimum size=10pt] (TAb) [right of=TAa] {f};
			\node[state,inner sep=1pt,minimum size=10pt] (TAc) [below of=TAb] {g};
			\node[state,inner sep=1pt,minimum size=10pt] (TAd) [below of=TAa] {h};

			\node[state,inner sep=1pt,minimum size=10pt] at (8, 1.5) (TBa) {i};
			\node[state,inner sep=1pt,minimum size=10pt] (TBb) [below right of=TBa] 
				{j};
			\node[state,inner sep=1pt,minimum size=10pt] (TBc) [below left of=TBa] 
				{k};

			\node[state,inner sep=1pt,minimum size=10pt] at (7.5, -2.5) (TCa) {l}; 	
			\node[state,inner sep=1pt,minimum size=10pt] (TCb) [right of=TCa] {m};

			\node[state,inner sep=1pt,minimum size=10pt] at (4.5, -2.5) (TD) {n};

			\draw [thick, rounded corners=.3cm] (-0.5,0.5) rectangle (1.5,-1.5); 
			\path (A) edge (B) (B) edge (C) (C) edge (D) (D) edge (A) edge (B);
			\node at (0.5, -1.8) {$S$};

			\path (TAa) edge (TAb) (TAc) edge (TAb) (TAd) edge (TAa) edge (TAc);
			\draw [thick, rounded corners=.3cm] (3.4,2.1) rectangle (5.6,-0.1); 
			\node at (4.5, 2.3) {$T_a$}; 
			\draw [thick, rounded corners=.2cm] (4.5, 1.9) rectangle (5.4, 0); 
			\draw[-,thick] (5.35, 0.05) -- (5.6, -0.3); 
			\node at (5.6, -0.5) {$U_{a, b}$};

			\path (TBa) edge (TBb) (TBb) edge [bend right] (TBc) (TBc) edge 
				[bend right] (TBb) edge (TBa); 
			\draw [thick, rounded corners=.3cm] (6.8,2) rectangle (9.2,0); 
			\node at (8, 2.2) {$T_b$}; 
			\draw [thick, rounded corners=.2cm] (8.3,1.2) rectangle (9.1, 0.4); 
			\node at (8.7, 1.4) {$U_{b, c}$};

			\path (TCa) edge [bend right] (TCb) (TCb) edge [bend right] (TCa); 
			\draw [thick, rounded corners=.3cm] (7,-1.5) rectangle (9,-3.5); 
			\node at (8, -3.8) {$T_c$}; 
			\draw [thick, rounded corners=.2cm] (7.15,-2.05) rectangle (8,-2.95);
			\node at (7.6, -3.2) {$U_{c, d}$};

			\path (TD) edge [loop left] (TD); 
			\draw [thick, rounded corners=.3cm] (3.5,-1.5) rectangle (5.5,-3.5); 
			\node at (4.5, -3.8) {$T_d$}; 
			\draw [thick, rounded corners=.2cm] (3.8,-2.1) rectangle (5,-2.9); 
			\node at (4.4, -3.2) {$U_{d, a}, U_{d, b}$};

			\path (TAb) edge (TBa) (TAc) edge [bend left] (TBc); 
			\path (TBb) edge (TCb); 
			\path (TCa) edge (TD); 
			\path (TD)  edge (TAc) edge (TAd) edge (TBc);

			\draw [thick, rounded corners=.5cm] (3, -4.2) rectangle (9.5, 2.7);
			\node at (6.25, -4.5) {$T$};
		\end{tikzpicture}
	\end{center}
	\caption{Interaction graphs of the modules detailed in Example~\ref{ex:2}. The 
	interaction graph of the original BAN is on the left and the interaction graph of 
	the simulating BAN is on the right. The simulating BAN is decomposed into four 
	sub-modules, one for each node in $S$. Notice that we did not represent the input 
	sets $E_a$, $E_b$, $E_c$ and $E_d$. The connections between the sets $T_a$, 
	$T_b$, $T_c$ and $T_d$ are based upon the interfaces defined in the example.}
	\label{figSim}
\end{figure}
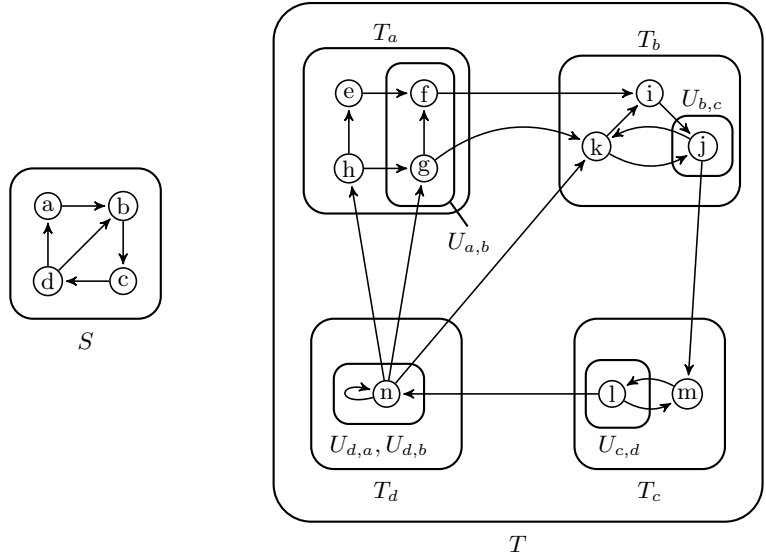

\begin{example} 
\label{ex:2}
	Let $S = \{a, b, c, d\}$. Let $T_a = \{e, f, g, h\}$, $T_b = \{i, j, k\}$, $T_c = 
	\{l, m\}$ and $T_d = \{n\}$. Let $T = T_a \cup T_b \cup T_c \cup T_d$. Let $E_a = 
	\{e_g, e_h\}$, $E_b = \{e_i, e_k, e'_k\}$, $E_c = \{ e_m \}$ and $E_d = 
	\{ e_n \}$. Let $U_{a, b} = \{f, g\}$, $U_{b, c} = \{j\}$, $U_{c, d} = \{l\}$, 
	$U_{d, a} = U_{d, b} = \{n\}$, and any other $U$ set empty. We will define 
	interfaces as the following: $I_{a, b}(e_i) = f$, $I_{a, b}(e_k) = g$, $I_{b, c}
	(e_m) = j$, $I_{c, d}(e_n) = l$, $I_{d, a}(e_h) = n$, $I_{d, a}(e_g) = n$ and 
	$I_{d, b}(e'_k) = n$ (see Figure~\ref{figSim}).
\end{example}

\begin{definition}
	Let $A$ be a set. A \emph{Boolean encoding} over A is a function $\phi: ( A \to 
	\B ) \to ( \{0, 1, \bullet\} )$, such that there exists at least one $x$ such 
	that $\phi(x) = 0$ and one $x$ such that $\phi(x) = 1$.
\end{definition}

For $x: A \to \B$ (a Boolean configuration over a set $A$), $\phi(x) = 1$ means 
that $x$ encodes a $1$, $\phi(x) = 0$ means 
that $x$ encodes a $0$, and $\phi(x) = \bullet$ means that $x$ does not encode any 
value. Each $\phi_a$ is defined as an encoding over $T_a$, and each $\phi_{a, b}$ as 
an encoding over $U_{a,b}$.

By definition we enforce that
$$\text{if } U_{a, b} \neq \emptyset \text{, then } \phi_a(x\ress{T_a}) \neq \bullet \implies \phi_{a, b}(x\ress{U_{a, b}}) = \phi_a(x\ress{T_a}).$$

Given a BAN on $S$ and some $a \in S$, let us now define the local simulation of
function $f_a$ by a module $M_a$. We want to express that given any configuration 
$x:S \to \B$, all the configurations $x':T_a \to \B$ and input configurations $i':E_a 
\to \B$ such that $x',i'$ encode the same information as $x$, the result of the 
dynamics on $x',i'$ in the simulating module must encode the result of the dynamics 
on $x$ in the simulated automaton. To express that $x'$ encodes the state of $a$ in 
$x$ is easy: $\phi_a(x')=x_a$. To express that $i'$ encodes the state of all $b \neq 
a$ in $x$ requires an additional notation. On the one hand we have $\phi_{b,a}: 
(U_{b,a} \to \B) \to (\{0, 1, \bullet\})$, and on the other hand we have $i':E_a \to 
\B$ describing the input-configuration of module $M_a$, and $I_{b,a}: E_a \to 
U_{b,a}$ describing the interface from $b$ to $a$. To plug these objects together, we 
put forward the hypothesis that if $I_{b, a}(e) = I_{b, a}(e')$, then $i'(e) = 
i'(e')$ for any $e, e' \in E_a$. This hypothesis is justified by the fact that the 
wiring applied by $I_{b, a}$ enforces the value of two inputs connected to the same 
element to be the same. Now, we define $i' \circ I_{b, a}^{-1}$ the configuration 
over $U_{b, a}$ such that $i' \circ I_{b, a}^{-1}(s) = i'(e)$ for any $e$ such that
$I_{b, a}(e) = s$. By our hypothesis this configuration is well defined.

\begin{definition}
	Let $a \in S$, $f_a$ be a Boolean function over $S$ and $M_a$ a module over 
	$(T_a, E_a, \alpha_a)$, with $\phi_a$ (resp. $\phi_{b,a}$) a Boolean encoding
	over $T_a$ (resp. $U_{b,a}$). Given a finite update mode $\Delta$ over
	$T_a$, $M_a$ \emph{locally simulates} $f_a$, denoted by $M_a \sim_\Delta f_a$, if 
	for all $x: S \to \B$,
	\begin{enumerate}
	\item and for all $x': T_a \to \B$ such that $\phi_a(x') = x_a$,\
	\item and for all $i': E_a \to \B$ such that for all $b \neq a$ we have 
		$\phi_{b, a} ( i' \circ I_{b, a}^{-1} ) = x_b$,
	\item we have:
		\begin{equation*}
			\phi_a( {M_a}_\Delta(x' \sqcup i') ) = f_a(x)\text{.}
		\end{equation*}
	\end{enumerate}
\end{definition}

This local simulation can be defined on a wide range of update modes $\Delta$. To 
ensure that the simulation works as planned at the global scale, we restrict the 
range of update modes $\Delta$ used for the local simulations, to those where no 
automata with input(s) are updated later than the first update.

\begin{definition}
	An update mode $\Delta$ over a module $M$ is defined to be \emph{input-first} if 
	for all $k > 1$ and all $s \in \Delta_k$, we have $\alpha(s) = \emptyset$.
\end{definition}

\begin{definition}
	We define that $M$ is able to \emph{input-first simulate} $f$ if there exists an 
	input-first $\Delta$ such that $M \sim_\Delta f$.
\end{definition}

Intuitively, such update modes let us make parallel the computation of modules; all 
information between modules is communicated simultaneously at the first frame of 
computation (update), followed by isolated updates in each module. To define global 
simulation, we introduce the global encoding $\Phi: (S \to \B) \to (S' \to \B) \cup 
\{ \bullet \}$ which always verifies that for all $x' : S' \to \B$, there exists
$x:S \to \B$ such that $\Phi(x) = x'$.

\begin{definition}
	Let $F$ and $F'$ be two Boolean automata networks over $S$ and $S'$ respectively.
	We define that $F$ \emph{simulates} $F'$, denoted by $F \sim F'$, if there exists 
	a global encoding $\Phi$ such that for all $x'$, $x$ such that $\Phi(x) = x'$, 
	and for all $\delta' \subseteq S'$, there exists a finite update mode $\Delta$ 
	over $S$ such that $\Phi( F_\Delta (x) ) = F'_{\delta'}(x')$.
\end{definition}

Given the definitions of local and global simulation, for any BAN $F$ over a set
$S$, we define each module $M_a$ as earlier, each defined over $(T_a, E_a,
\alpha_a)$, along side each set $U_{a, b}, I_{a, b}$ and each encoding
$\phi_a, \phi_{a, b}$.
\begin{theorem}
\label{th-simulation}
	Let $F$ be a BAN over $S$. For each $a \in S$, let $M_a$ be a module over $(T_a, 
	E_a, \alpha_a)$ that locally simulates $F(a)$ in an input-first way. There exists 
	a recursive wiring $\omega$ over $T=\bigcup_{a \in S}T_a$ such that
	\begin{equation*}
		\circlearrowright_\omega \left( \bigWiring_{a \in S} M_a \right) \sim F\text{.}
	\end{equation*}
\end{theorem}
\begin{sproof}
We prove that the execution of the module $M$ obtained from the wiring $\omega$
can be built from the execution of each $M_a$. We apply the hypothesis of local
simulation on each $M_a$, and obtain a global simulation.
\end{sproof}

This theorem helps us investigate if every BAN can be simulated by a BAN with
a given property, hence justifying that theoretical studies can impose some restrictions without loss of generality.
If every function $f$ can be locally simulated by a given
module with a property $\mathcal{P}$, and if property $\mathcal{P}$ is preserved over 
wirings, then we know that any BAN can be simulated by another BAN with the property 
$\mathcal{P}$. This is formally proven for the following cases.

\begin{corollary}
Let $F$ be a BAN. There exists $F'$ such that $F' \sim F$ and every function of
$F'$ is a disjunctive clause.
\end{corollary}

\begin{corollary}
Let $F$ be a BAN. There exists $F'$ such that $F' \sim F$ and every function of
$F'$ is monotone.
\end{corollary}

\begin{sproof}
Both of theses results are obtained by replacing the automata of $F$ by modules
that locally simulates them. For disjunctivity, the module has one automaton for
each clause of the conjunctive normal form of the simulated function, and one 
for the result
(using De Morgan's law we convert the outer conjunction to a disjunction). For 
monotony, we use a lemma that shows that we can always construct a monotone
function from any function at the cost of duplicating each variable. Using this
lemma we construct a network with twice the automata which locally simulates any
function. The results are obtained by the Theorem \ref{th-simulation}.
\end{sproof}

It can seem strange that this particular theorem applies to BANs and not
to modules (as it would be a more general result). Such a result would
need a definition of simulation between modules, and such a definition
would imply an interpretation of the information provided by the
simulating module's inputs. We choose not to develop this particular idea, 
as this theorem was only meant to
apply to BANs, but a generalisation of this result to modules
would be
a good subject for future works.

\section{Examples}
\label{s:examples}

To illustrate and justify the notions that are presented in Section~\ref{s:motivations},
we shall now present two examples of BANs that can be partially understood by cutting them into modules. The
first example is a toy BAN illustrated in Figure~\ref{fig:artificial}.
In this representation we assume the function of each automaton to be a
disjunctive clause with one literal for each incident edge, the sign of which
dictates the sign of the literal.

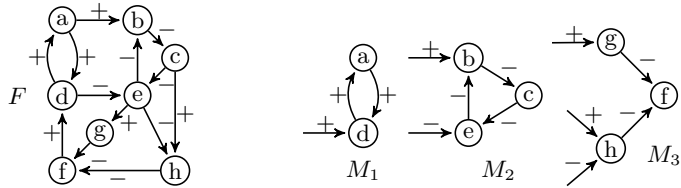
\begin{figure}[t!]
\begin{center}
\begin{tikzpicture}[->,>=stealth',shorten
>=1pt,auto,node distance=1cm, semithick, initial text=,inner sep=0pt, minimum
size=0pt]
	\node[state,inner sep=1pt,minimum size=10pt] (A) at (0, 0) {a};
	\node[state,inner sep=1pt,minimum size=10pt] (B) [right of=A] {b};
	\node[state,inner sep=1pt,minimum size=10pt] (C) at (1.5, -0.5) {c};
	\node[state,inner sep=1pt,minimum size=10pt] (D) [below of=A] {d};
	\node[state,inner sep=1pt,minimum size=10pt] (E) [below of=B] {e};
	\node[state,inner sep=1pt,minimum size=10pt] (F) [below of=D] {f};
	\node[state,inner sep=1pt,minimum size=10pt] (G) at (0.5, -1.5) {g};
	\node[state,inner sep=1pt,minimum size=10pt] (H) at (1.5, -2) {h};
	\path (A) edge [bend left] node {+} (D)
	      (D) edge [bend left] node {+} (A)
	      (A) edge node {$+$} (B)
	      (D) edge node {$-$} (E)
	      (B) edge node {$-$} (C)
	      (C) edge node {$-$} (E)
	      (E) edge node {$-$} (B)
	      (E) edge node {$+$} (G)
	      (E) edge node {$-$} (H)
	      (C) edge node {$+$} (H)
	      (G) edge node {$\ -$} (F)
	      (H) edge node {$-$} (F)
	      (F) edge node {$+$} (D);

	\node[state,inner sep=1pt,minimum size=10pt] (A') at (4, -0.5) {a};
	\node[state,inner sep=1pt,minimum size=10pt] (B') at (5.4, -0.5) {b};
	\node[state,inner sep=1pt,minimum size=10pt] (C') at (6.2, -1) {c};
	\node[state,inner sep=1pt,minimum size=10pt] (D') [below of=A'] {d};
	\node[state,inner sep=1pt,minimum size=10pt] (E') [below of=B'] {e};
	\node[state,inner sep=1pt,minimum size=10pt] (F') at (8, -1) {f};
	\node[state,inner sep=1pt,minimum size=10pt] (G') [above left of=F'] {g};
	\node[state,inner sep=1pt,minimum size=10pt] (H') [below left of=F'] {h};
	\path (A') edge [bend left] node {+} (D')
	      (D') edge [bend left] node {+} (A')

	      (B') edge node {$-$} (C')
	      (C') edge node {$-$} (E')
	      (E') edge node {$-$} (B')

	      (G') edge node {$-$} (F')
	      (H') edge node {$-$} (F');
	
	\draw (3.2, -1.5) -- node {$+$} (D'); 

	\draw (4.6, -0.5) -- node {$+$} (B'); 
	\draw (4.6, -1.5) -- node {$-$} (E'); 

	\draw (6.51, -0.3) -- node {$+$} (G'); 
	\draw (6.71, -2.2) -- node {$-$} (H');
	\draw (6.71, -1.2) -- node {$+$} (H');

	\node at(-0.6, -1) {$F$};
	\node at(4, -2) {$M_1$};
	\node at(5.8, -2) {$M_2$};
	\node at(8, -1.8) {$M_3$};
\end{tikzpicture}
\end{center}
\caption{Representation of a handmade Boolean automata network $F$ next to
the three different modules $M_1$, $M_2$ and $M_3$ that compose it. The
function of each automaton is defined as a disjunctive clause with a positive
literal for each incident ``$+$'' edge, and a negative literal for each
incident ``$-$'' edge. For example, $f_h(x) = x_c \vee \neg x_e$.}
\label{fig:artificial}
\end{figure}

Looking at this example, it does not seem easy to express the entire behaviour of the
BAN $F$. Its representation is a strongly connected graph with
multiple interconnected positive and negative cycles. Yet, cutting this graph
into multiple modules and analysing the functionality of each of them
is an easy way to understand interesting parts of the dynamics of the network.

By assuming the decomposition of $F$ as shown in Figure~\ref{fig:artificial},
we can start to attach a functionality to each module. Module $M_1$ is a positive
cycle, where the configuration $x_a = x_d = 1$ is a fixed point (whatever the input). Its
functionality
can be identified as a ``one time button'' that cannot be pushed back. Module
$M_2$ is a negative cycle, which are known for their long limit cycles. The
difference here is that as $M_2$ has two inputs, its behaviour can be stabilised
into a fixed point by a fixed input. For example, the fixed point $x_b=x_e=1, x_c=0$
can be obtained with the constant input $i_b=1, i_e=0$. Finally, the module $M_3$
is acyclic and thus only computes the Boolean function
$\neg i_g \vee ( \neg i_h \wedge i_{h'} )$. It follows that $M_3$ stabilises
to a fixed point under any constant input.

This simple analysis leads us to the following conclusion : every fair execution
(meaning executing every automaton an infinite amount of time) of $F$ which
verifies $x_a = x_d = 1$ at any moment stabilises into a fixed point. This is true
because $x_a = x_d = 1$ implies that the ``one time button'' of $M_1$ is pushed in,
which locks the behaviour of $M_2$ into a fixed point, which leads $M_3$ to
compute a Boolean function over a fixed input. This somewhat informal
demonstration has led us to a conclusion that was not easily implied by
the architecture of the network, showcasing the usefulness of understanding
networks as composition of parts to which one can assign functionalities.

The second example is drawn from a model predicting the cell cycle sequence
of fission yeast~\cite{S-Davidich2008}. This network
is represented in Figure~\ref{fig:yeast-ban}, and can be
decomposed into a more abstract network, where each node represents a module
of the original network. This network is represented in Figure~\ref{fig:yeast-abstract}
and its modules are constructed as follows:
$C = \{Rum1, Ste9\},
D = \{Cdc, Cdc*\}, F = \{Cdc25\}, G = \{Mik\}, I = \{Start, SK\},
J = \{PP, Slp1\}$. A quick analysis of these modules leads us to sort them
into three categories : cycles ($C, D$), functions ($F, G$) and
igniters ($I, J$). Let us now explain this organisation in an informal way.

The two cycle modules $C$ and $D$ are organised in a 4-cycle of negative
feedback which means that if considered separately from the rest of the network,
those two modules would behave as antagonists: in most cases, when the
automata of $C$ (resp. $D$) are evaluated to $1$, the automata of $D$
(resp. $C$) will be evaluated to $0$.
Modules $F$ and $G$ can be viewed as functions which help $D$ and $C$ respectively to
be evaluated to $1$; they both are influenced by $J$ in different ways. Modules $I$ and
$J$ are called igniters because they turn themselves to $0$ every time they
are evaluated to $1$, but not before influencing the other nodes. Module $I$ inhibits
$C$ when activated, and can be considered as the input of the whole network.
Module $J$ is activated by $D$, activates $C$ and $G$, and inhibits $F$.

From this we can conclude that if the network stabilises, it will more likely
stabilise by evaluating $C$ to $1$ and $D$ to $0$. This conclusion arises from
the fact that $D$ activates $J$, which in turn inhibits $D$ directly, but also
inhibits $F$ (which activates $D$) and activates $G$ (which inhibits $D$).
This also means that $F$ will be evaluated to $0$ and $G$ to $1$. Finally,
$I$ and $J$ will naturally be evaluated to $0$ because of the natural negative
feedback that compose them. This particular evaluation of the network (only $C$
and $G$ to $1$) is actually the main fixed point of the network's dynamics
put forward in~\cite{S-Davidich2008} and
is named $G1$. This shows that such a fixed point can be described without the
need to compute the $2^{10} = 1024$ different configurations of the network
and their dynamics.

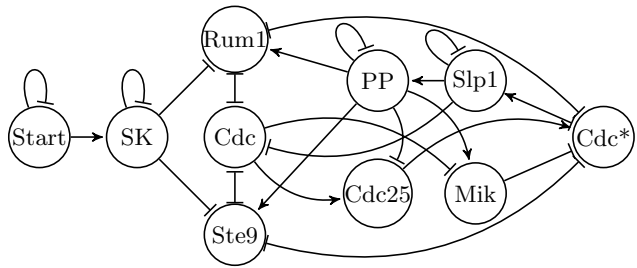
\begin{figure}[t!]
\begin{center}
\begin{tikzpicture}[->,>=stealth',shorten
>=1pt,auto,node distance=1.3cm, semithick, initial text=,inner sep=0pt, minimum
size=0pt]
	\node[state] (Start) at (0, 0) {Start};
	\node[state] (SK) [right of=Start] {SK};
	\node[state] (Cdc) [right of=SK] {Cdc};
	\node[state] (Rum) [above of=Cdc] {Rum1};
	\node[state] (Ste) [below of=Cdc] {Ste9};
	\node[state] (PP) at (4.5, 0.75) {PP};
	\node[state] (Slp) [right of=PP] {Slp1};
	\node[state] (Cdc25) at (4.5, -0.75) {Cdc25};
	\node[state] (Mik) [right of=Cdc25] {Mik};
	\node[state] (Cdc*) at (7.5, 0) {Cdc*};

	\path (Rum) edge [|-|, bend left] (Cdc*)
	      (Ste) edge [|-|, bend right] (Cdc*)
	      (Rum) edge [|-|] (Cdc)
	      (Ste) edge [|-|] (Cdc);

	\path (Start) edge (SK)
              (SK) edge[-|] (Rum)
              (SK) edge[-|] (Ste);

	\draw[-|] (Start) to [out=110,in=80,looseness=8] (Start);
	\draw[-|] (SK) to [out=110,in=80,looseness=8] (SK);

	\path (Slp) edge [<-|] (Cdc*)
		    edge (PP)
		    edge [-|, bend left] (Cdc);
	\draw[-|] (Slp) to [out=150,in=120,looseness=8] (Slp);

	\path (PP) edge (Rum)
		   edge (Ste)
		   edge [bend left] (Mik)
		   edge[-|, bend left] (Cdc25);
	\draw[-|] (PP) to [out=140,in=110,looseness=8] (PP);

	\path (Mik) edge [-|] (Cdc*);

	\path (Cdc25) edge[bend left] (Cdc*);

	\path (Cdc) edge [bend right] (Cdc25)
		    edge [-|, bend left] (Mik);

\end{tikzpicture}
\end{center}
\caption{Representation of the network simulating the cell cycle sequence of
fission yeast extracted from~\cite{S-Davidich2008}. Activating interactions
are represented by simple arrows and inhibiting interactions by flat arrows.
The detail of each node's function is available in the original paper.}
\label{fig:yeast-ban}
\end{figure}

\begin{figure}[t!]
\begin{center}
\begin{tikzpicture}[->,>=stealth',shorten
>=1pt,auto,node distance=1cm, semithick, initial text=,inner sep=0pt, minimum
size=0pt]
	\node[state,inner sep=1pt,minimum size=10pt] (C) at (0, 0) {C};
	\node[state,inner sep=1pt,minimum size=10pt] (D) [right of=C] {D};
	\node[state,inner sep=1pt,minimum size=10pt] (G) [right of=D] {G};
	\node[state,inner sep=1pt,minimum size=10pt] (I) [below of=C] {I};
	\node[state,inner sep=1pt,minimum size=10pt] (J) [right of=I] {J};
	\node[state,inner sep=1pt,minimum size=10pt] (F) [right of=J] {F};

	\path (I) edge [-|] (C);
	\draw[-|] (I) to [out=215,in=155,looseness=8] (I);

	\path (C) edge [|-|] (D);

	\path (D) edge [|-|] (G);
	\path (D) edge [|-|] (G);

	\path (J) edge (C);
	\path (J) edge [<-|] (D);
	\path (J) edge (G);
	\path (J) edge [-|] (F);
	\draw[-|] (J) to [out=210,in=170,looseness=8] (J);

	\path (F) edge [<->] (D);

\end{tikzpicture}
\end{center}
\caption{Abstract representation of the interactions between the modules
$C, D, F, G, I$ and $J$ based upon the network represented in
Figure~\ref{fig:yeast-ban}.}
\label{fig:yeast-abstract}
\end{figure}
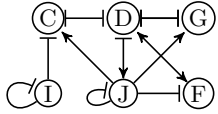

\section{Conclusion}

The two theorems formulated in this article tell us that seeing BANs as modular 
entities is a way to discover useful results. With the simple addition of inputs to 
BANs, we have expressed a general simulation structure that can be used to understand 
the computational nature and limits of given properties over BANs. Let us underline 
that all the definitions and results can be applied to BANs and modules defined over 
countably infinite sets of automata and inputs. 

Wherever Turing-completeness is observed, complex behaviours emerge that cannot
be simply or quickly formulated from the basic rules of the computation.
In such situations, the solution is either to compute every single possibility
to capture the whole dynamics of the observed system, or to simplify the
model. We believe that the framework developed in this paper is a strong
candidate to enable us to decompose complex networks into parts with tractable
functionalities, and to make conclusions about the whole network at a
cheaper cost. This approach is still very informal at this moment and will be 
the focus of further developments.

\paragraph{Acknowledgements} This work has been supported ``Investissement d'avenir" program ANR-16-CONV-00001 and PACA Project Fri 2015\_01134.

\bibliographystyle{plain}
{\small{\bibliography{stage}}}

\appendix

\section{Proofs}

\setcounter{property}{0}
\setcounter{theorem}{0}
\setcounter{equation}{0}
\setcounter{corollary}{0}

\begin{property}
	The following statements hold.
	\begin{enumerate}[\em(i)\quad]
	\item $\forall M,\quad \circlearrowright_\emptyset M = M$.
	\item $\forall M, M',\quad M \rightarrowtail_\emptyset M' = M' 	
		\rightarrowtail_\emptyset M$.
	\item $\forall M, M', M'',\quad M \rightarrowtail_\emptyset 
		( M' \rightarrowtail_\emptyset M'' ) = ( M \rightarrowtail_\emptyset M' ) 
		\rightarrowtail_\emptyset M''$.
	\end{enumerate}
\end{property}

\begin{proof}
	\begin{equation*}
		\forall M, M', M \rightarrowtail_\emptyset M' = M' 
			\rightarrowtail_\emptyset M\text{.}
	\end{equation*}
	By definition, $M \rightarrowtail_\emptyset M'$ and $M' \rightarrowtail_\emptyset 
	M$ are both defined on $(S \cup S', E \cup E', \alpha \sqcup \alpha')$. For any 
	$s \in S$, $M \rightarrowtail_\emptyset M'(s) = M' \rightarrowtail_\emptyset 
	M(s)$ and for $s' \in S'$, $M \rightarrowtail_\emptyset M'(s') = M' 
	\rightarrowtail_\emptyset M(s')$.
	\begin{equation*}
		\forall M, \circlearrowright_\emptyset M = M\text{.}
	\end{equation*}
	By a similar argument, $\circlearrowright_\emptyset M$ is by definition defined 
	on $(S, E, \alpha)$ such that $\circlearrowright_\emptyset M(s) = M(s)$ for any 
	$s \in S$.
	\begin{equation*}
		\forall M, M', M", M \rightarrowtail_\emptyset ( M' \rightarrowtail_\emptyset 
		M" ) = ( M \rightarrowtail_\emptyset M' ) \rightarrowtail_\emptyset M"\text{.}
	\end{equation*}
	By definition, the left side of this equation is defined over 
	$(S \cup S' \cup S", E \cup E' \cup E", \alpha \sqcup \alpha' \sqcup \alpha")$ as 
	is the right side of this equation. The two modules defining the same functions, 
	we obtain the result.\qed
\end{proof}

\begin{theorem}
	\label{th1}
	Let $M$ be a module and $\{M_p \mid p \in P\}$ a partition of that module. There 
	exists a recursive wiring $\omega$ such that
	\begin{equation}
		M =\ \circlearrowright_{\omega} \left( \bigWiring_{p \in P} M_p \right)\text{.}
	\end{equation}
\end{theorem}

\begin{proof}
	By definition of the empty wiring, the module $\bigWiring_{p \in P} M_p$ is 
	defined over $(S, E \cup \bigcup_{p \in P} \tau_p(S \setminus S_p), \bigsqcup_{p 
	\in P} \alpha_p)$ and for all $s \in S$, $x : S \rightarrow \B$ and $i : E 
	\rightarrow \B$ verifies
	\begin{equation}
	\label{th1:eq1}
		\left(\bigWiring_{p \in P} M_p\right) (s) (x \sqcup i') = 
			M(s)(x \sqcup i)\text{.}
	\end{equation}
	Knowing that $i'(e) = i(e)$ for $e \in E_s$, and $i'(s) = x( \tau^{-1}_p(s) )$ 
	for $s \in Q_p$. Let $\omega$ be the recursive wiring over $\bigWiring_{p \in P} 
	M_p$ with domain $\bigcup_{p \in P} \tau_p(S \setminus S_p)$ such that $\omega(q) 
	= \tau_p^{-1}(q)$ given p such that $q \in Q_p$.\smallskip

	\noindent By definition of the recursive wiring, the module 
	$\circlearrowright_\omega (\bigWiring_{p \in P} M_p)$ is defined over the set $(S, 
	E, \alpha)$. For all $s, x, i$, we now have that
	\begin{equation}
	\label{th1:eq2}
		\circlearrowright_\omega \left(\bigWiring_{p \in P} M_p \right) (s) 
		(x \sqcup i) =  \left(\bigWiring_{p \in P} M_p\right)(s)(x \sqcup i\ress{E_s} 
		\sqcup (x \circ \omega\ress{\tau_p(S \setminus S_p)}))\text{.}
	\end{equation}
	By our definitions of $\omega$ and $i'$, we have that $i' = i\ress{E_s} \sqcup (x 
	\circ \omega\ress{\tau_p(S \setminus S_p)})$. From that, and 
	Equations~\ref{th1:eq1} and \ref{th1:eq2}, we infer that for all $s, x, i$:
	\begin{equation*}
		\circlearrowright_\omega \left(\bigWiring_{p \in P} M_p \right) (s) (x \sqcup 
		i) = M(s)(x \sqcup i)\text{.}
	\end{equation*}
	Therefore for any $s$:
	\begin{equation*}
		\circlearrowright_\omega \left(\bigWiring_{p \in P} M_p \right) (s) = 
			M(s)\text{,}
	\end{equation*}
	which concludes the proof. \qed
\end{proof}

\begin{corollary}
	The set of all modules is equal to the closure by any wiring of the set of 
	modules of size $1$ :
	\begin{equation*}
		\mathcal{M} = \overline{\mathcal{M}}_1^\omega\text{.}
	\end{equation*}
\end{corollary}

\begin{proof}
	Trivially, $\overline{\mathcal{M}}_1^\omega \subseteq \mathcal{M}$. For any $M \in 
	\mathcal{M}$ of size $n$, we know by Theorem~\ref{th1} that in particular the 
	$n$-partition of $M$ into sub-modules of size $1$ can be wired into the original 
	module $M$. Therefore $\mathcal{M} = \overline{\mathcal{M}}_1^\omega$.\qed
\end{proof}

\begin{theorem}
\label{th2}
	Let $F$ be a BAN over $S$. Let $\{M_a \mid a \in S\}$ be a set such that for 
	every $a$, $M_a$ is a module over $(T_a, E_a, \alpha_a)$ that simulates $F(a)$ in 
	a input-first way. There exists $\omega$ a recursive wiring over $T$ such that :
	\begin{equation*}
		\circlearrowright_\omega \left( \bigWiring_{a \in S} M_a\right) \sim F\text{.}
	\end{equation*}
\end{theorem}

\begin{proof}
	By definition of the empty wiring, $\bigWiring_{a \in S} M_a$ is defined over 
	$(T, \bigcup_{a \in S} E_a,$ $\bigsqcup_{a \in S} \alpha_a)$. Let $\omega = 
	\bigcup_{a, b \in S, a \neq b} I_{a, b}$. By definition of $I_{a, b}$, we can 
	easily see that the module $M =\ \circlearrowright_\omega \left( \bigWiring_{a \in 
	S} M_a\right)$ is defined over $(T, \emptyset, s \mapsto \emptyset)$ and can be 
	seen as a Boolean automata network. Let us prove that, for all $a \in S$, for all 
	input-first simulating update mode $\Delta$ for the module $M_a$, for any 
	$\Delta'$ update mode over $T \setminus T_a$, and for any $x : T \rightarrow \B$, 
	the following equation holds:
	\begin{equation}
	\label{th2:eq1}
		M_{\Delta \cup \Delta'} (x)\ress{T_a} = {M_a}_\Delta ( x\ress{T_a} \sqcup (x 
			\circ \bigsqcup_b I_{b, a}) )\text{.}
	\end{equation}
	At the first step of the execution, the wiring $\omega$ implies that for any $s 
	\in T_a$, for any $x$, $M(s)(x)$ $=$ $\left( \bigWiring_{a \in S} M_a \right)(s)
	(x \sqcup (x \circ \omega))$. 
	From the definition of the empty wiring, we can deduce in particular that $M(s)
	(x)$ $=$ $M_a(s)($ $x\ress{T_a}\sqcup (x \circ \omega\ress{E_a}))$. By definition 
	of the interfaces, this notation is equivalent to $\forall s \in T_a, M(s)(x) = 
	M_a(s)(x\ress{T_a} \sqcup (x \circ \bigsqcup_{b} I_{b, a}))$.\smallskip

	Let us define $A = \{ s \in T_a \mid \alpha(s) \neq \emptyset \}$ and $B = T_a 
	\setminus A$. By the definition of $\Delta$, we know that $s \in \Delta_k$ with 
	$k > 0$ implies $s \in B$.\\[2mm] 
	Let us look at the $A$ part of this problem. Let $\delta = \Delta_0$ and $\delta' 
	= \Delta'_0$. We can trivially deduce from the previous statement that:
	\begin{equation*}
		M_{\delta \cup \delta'}(x)\ress{A} = {M_a}_{\delta}(x\ress{T_a} \sqcup (x 
		\circ \bigsqcup_{b} I_{b, a}))\ress{A}\text{.}
	\end{equation*}
	Furthermore, there is no $s \in A$ such that $s \in \Delta_k$ for any $k > 0$. We 
	can simply conclude since no update is made to any function of $A$ in the rest of 
	the execution that $M_{\delta \cup \delta'}(x)\ress{A} = M_{\Delta \cup \Delta'}
	(x)\ress{A}$, and that ${M_a}_{\delta}(x\ress{T_a} \sqcup (x \circ \bigsqcup_{b} 
	I_{b, a}))\ress{A} = {M_a}_{\Delta}(x\ress{T_a} \sqcup (x \circ \bigsqcup_{b} 
	I_{b, a}))\ress{A}$. In conclusion of this $A$ part, $M_{\Delta \cup \Delta'}
	(x)\ress{A} = {M_a}_{\Delta}(x\ress{T_a}$ $\sqcup (x \circ \bigsqcup_{b} I_{b, 
	a}))\ress{A}$.\\[2mm]
	Let us now consider the $B$ part of the problem. For $s \in B$, we have $M(s)(x) 
	= M_a(s)(x\ress{T_a} \sqcup ( x \circ \omega\ress{E_s}))$. By definition of $B$, 
	$s \in B$ implies $E_s = \emptyset$. We can conclude that $\forall s \in B, M(s)
	(x) = M_a(s)(x\ress{T_a})$. We deduce, for any $\delta \subseteq T_a$ and 
	$\delta' \subseteq T \setminus T_a$, that $M_{\delta \cup \delta'} (x)\ress{B} = 
	{M_a}_\delta ( x\ress{T_a} \sqcup i )\ress{B}$, for $i$ any input configuration
	over $E_a$. By a simple recursive demonstration, we can easily show that 
	$M_{\Delta \cup \Delta'} (x)\ress{B} = {M_a}_\Delta ( x\ress{T_a} \sqcup i 
	)\ress{B}$.\\[2mm]
	Reuniting the $A$ and $B$ parts of this demonstration, we obtain that 
	$M_{\Delta \cup \Delta'} (x)$ $= {M_a}_{\Delta}(x\ress{T_a} \sqcup (x \circ 
	\bigsqcup_{b} I_{b, a}))\ress{A} \cup {M_a}_\Delta ( x\ress{T_a} \sqcup i 
	)\ress{B}$. Assuming $i = x \circ \bigsqcup_{b} I_{b, a}$, we obtain $M_{\Delta 
	\cup \Delta'} (x) = {M_a}_{\Delta}(x\ress{T_a} \sqcup (x \circ \bigsqcup_{b} 
	I_{b, a}))$, and prove the lemma described in Equation~\ref{th2:eq1}.\smallskip

	Let us now define $\Phi : ( T \rightarrow \B ) \rightarrow ( S \rightarrow \B ) 
	\cup \{ \emptyset \}$ such that, for any $x : T \rightarrow \B$, $\Phi(x) = 
	\emptyset$ if there exists $a \in S$ such that $\phi_a(x\ress{T_a}) = \emptyset$, 
	and $\Phi(x)(a) = \phi_a(x\ress{T_a})$ otherwise. Let $x$ and $x'$ such that 
	$\Phi(x) = x'$, and $x' \neq \emptyset$. Let $\delta \subseteq S$ be an update 
	over $F$. Let us define, for any $a \in \delta$, the update mode $\Delta_a$ such 
	that $\Delta_a$ is an input-first update mode upon which $M_a$ simulates the 
	function $F(a)$ ; by hypothesis such an update mode can always be found.\\[2mm]
	Let us define the update mode $\Delta$ over $T$ such that $\Delta = \bigcup \{ 	
	\Delta_a \mid a \in \delta \}$. We will now prove that $\Phi( M_\Delta(x) ) = 
	F_\delta(x')$. First, we can clearly see that $M_\Delta(x) = \bigsqcup \{ 
	M_\Delta(x)\ress{T_a} \mid a \in S \}$, which can be developed into 
	$M_\Delta(x) = \bigsqcup \{ M_\Delta(x)\ress{T_a} \mid a \in \delta \} \sqcup 
	\bigsqcup \{ x\ress{T_a} \mid a \in S \setminus \delta \}$, from which we infer:
	\begin{equation*}
		M_\Delta(x) = \bigsqcup \{ M_{\Delta_a \cup \bigcup_{b \in \delta, b \neq a} 
			\Delta_b}(x)\ress{T_a} \mid a \in \delta \} \sqcup \bigsqcup \{ 
			x\ress{T_a} \mid a \in S \setminus \delta \}\text{.}
	\end{equation*}
	Using the lemma formulated in Equation~\ref{th2:eq1}, this can be rewritten into:
	\begin{equation*}
		M_\Delta(x) = \bigsqcup_{a \in \delta} {M_a}_{\Delta_a}(x\ress{T_a} \sqcup (x 
		\circ \bigsqcup_{b} I_{b, a})) \sqcup \bigsqcup_{a \in S \setminus \delta} 
		x\ress{T_a}\text{.}
	\end{equation*}
	As the result of an execution of the module $M_a$ is always defined as a 
	configuration over $T_a$, we can infer the following encoding of $M_\Delta(x)$ by 
	$\Phi$ :
	\begin{equation*}
		\Phi( M_\Delta(x) )(a)=\begin{array}\{{ll}.
			\phi_a({M_a}_{\Delta_a}(x\ress{T_a} \sqcup (x \circ \bigsqcup_{b} 
				I_{b, a}))) & \text{ if } a \in \delta\\
			\phi_a(x\ress{T_a}) & \text{ if } a \in S \setminus \delta
		\end{array}\text{.}
	\end{equation*}
	We know by definition of $x$ and $x'$ that $\phi_a(x\ress{T_a}) = x'_a$ and that
	$\phi_{b, a} (x \circ I_{b, a} \circ I^{-1}_{b, a}) = \phi_{b, a} (x 
	\ress{U_{b,a}}) = \phi_b (x \ress{T_b}) = x'_b$ by definition of $\phi_{b, a}$. 
	From this we can apply the local simulation definition and obtain:
	\begin{multline*}
		\Phi( M_\Delta(x) )(a)=\begin{array}\{{ll}.
			f_a(\Phi(x)) & \text{ if } a \in \delta\\
			\phi_a(x\ress{T_a}) & \text{ if } a \in S \setminus \delta
		\end{array}\\
		\iff \Phi( M_\Delta(x) )(a)=\begin{array}\{{ll}.
			f_a(\Phi(x)) & \text{ if } a \in \delta\\
			\Phi(x)(a) & \text{ if } a \in S \setminus \delta.
		\end{array}\text{.}
	\end{multline*}
	Futhermore, by the definition of an update over $F$, we can write that:
	\begin{equation*}
		F_\delta(x')(a) =\begin{array}\{{ll}.
			f_a(x') & \text{ if } a \in \delta\\
			x'(a) & \text{ if } a \in S \setminus \delta
		\end{array}\text{.}
	\end{equation*}
	Finally, by definition of $x' = \Phi(x)$:
	\begin{equation*}
		F_\delta(x')(a) =\begin{array}\{{ll}.
			f_a(\Phi(x)) & \text{ if } a \in \delta\\
			\Phi(x)(a) & \text{ if } a \in S \setminus \delta
		\end{array}\text{,}
	\end{equation*}
	which implies $\Phi(M_\Delta(x)) = F_\delta(x')$, and concludes the proof.\qed
\end{proof}

\begin{corollary}
Let $F$ be a BAN. There exists $F'$ such that $F' \sim F$ and every function of
$F'$ is a disjunctive clause.
\end{corollary}

\begin{proof}
With Theorem \ref{th-simulation} in mind, we only need to demonstrate that for
any function $f$, there exists a module locally simulating it in a input-first
way, in which every function is a disjunctive clause.

Let us consider $F$ a BAN set over $S$. Let $a \in S$. We decompose $f_a$ into a 
set of disjunctive clauses $C_a$ such that $f_a(x) = \bigwedge\limits_{c \in
C_a} c(x)$.

Let $M_a = (T_a, E_a, \alpha_a)$ be a module with $T_a = \{u_c \mid c \in C\} 
\cup \{r_a\}$, $E_a = \{e_{b, c, a} \mid a \neq b \text{, and the variable } x_b
\text{ is included in clause } c\}$.
For all $b, c$, $e_{b, c, a} \in \alpha(u_c)$
if and only if $x_b$ is included in clause $c$. For $c \neq c'$, $e_{b, c, a} 
\notin \alpha(u_c')$ and $\alpha(r_a) = \emptyset$.

For $c \in C_a$, $x$ a configuration over $T_a$ and $e$ a configuration over
$E_a$, $M_a(u_c)$ is the function described by $f_{u_c}(x, e) = c(
x_a \mapsto \neg x(r_a) \sqcup x_b \mapsto \neg e( e_{b, c, a} ) )$. The function
$M(r)$ is the function $f_{r_a}(x, e) = \bigvee\limits_{c \in C_a} \neg x(u_c)$.

This local module is shaped as a pyramid where the base is constitued of
one node for every disjunctive clause of the simulated function, and the top
of exactly one node that represents the result of the function. It follows
from this definition that every function of this module is a disjunctive clause.
An illustrated example of such a local module is presented in
Figure~\ref{fig:local-disjunctive-module}.

\begin{figure}[t!] 
	\begin{center}
		\begin{tikzpicture}[->,>=stealth',shorten
>=1pt,auto,node distance=1.6cm, semithick, initial text=,inner sep=0pt, minimum
size=0pt] 
			\node[state] (R) {$r_a$};
			\node[state] (C) [above left of=R] {$u_c$}; 
			\node[state] (C') [below left of=R] {$u_{c'}$};

			\path (C) edge[bend left] node {$-$} (R); 
			\path (C') edge node {$-$} (R); 

			\path (R) edge[bend left] node {$-$} (C);

			\draw (-2, -0.73) -- node {$+$} (C');
			\draw (-2, -1.53) -- node {$-$} (C');

			\node at (-2.5, -0.73) {$e_{b, c', a}$};
			\node at (-2.5, -1.53) {$e_{d, c', a}$};

			\draw (R) -- (1, 0);
			\node at (1.6, 0) {$\neg f_a(x)$};

		\end{tikzpicture}
	\end{center} 
	\caption{Interaction graph of the locally disjunctive module for the
example function $f_a(x) = x_a \wedge ( \neg x_b \vee x_d )$. We name the
clauses of $f_a$ as $c = x_a$ and $c' = \neg x_b \vee x_d$. Notice that most
of the signs are inversed to simulate a AND gate.}
	\label{fig:local-disjunctive-module}
\end{figure}
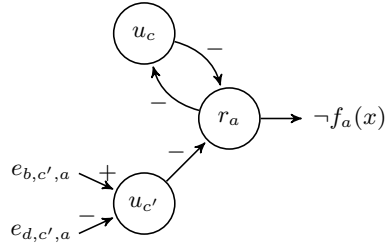

We define $U_{b, a}$ such that $U_{b, a} = \{r_a\}$ if the variable $x_b$ is 
included in one of the clauses of the function $f_a$, and $U_{b, a} = \emptyset$
otherwise.

The encodings $\phi_a$ and $\phi_{b,a}$ for every $b$ such that $U_{b, a} \neq 
\emptyset$ are defined such that $\phi_a(x) = \phi_{b, a}(x\ress{U_{b,a}}) = 
\neg x(r_a)$. This means that the node $r$ represents the inverse of the result
of the function.

We always define $I_{b, a}(e_{b, c, a}) = r_b$. More intuitively, to resolve
the value of the variable $x_b$ in a clause of $f_a$, look for the value of
the node $r_b$ in the local module $M_b$. We reverse it back to the correct
value thanks to the inversion of each input of each clause automaton.

\begin{lemma}
$M_a$ locally simulates $f_a$ in a input-first way.
\end{lemma}

Let $\Delta_a = (\{u_c \mid c \in C_a\}, \{r_a\})$ be an input-first update mode 
for the module $M_a$. We will sometimes note $\Delta_a = (\delta, \delta_r)$
in further developments.

Let $x$ be a configuration over $F$. Let $x'$ be a configuration over $T_a$
such that $\phi_a(x') = x_a$. Let $i'$ be an input configuration over $E_a$
such that for any $b \neq a$, $\phi_{b, a}(i' \circ I^{-1}_{b, a}) = x_b$.

Such a $x'$ is a configuration over $T_a$ with $x'(r_a) = \neg x_a$.
Such a $i'$ is a configuration over $E_a$ such that $i'(e_{b, c, a}) = \neg x_b$
for
every $b$ and $c$. Such configurations are well defined and can always be found.

To prove the above lemma, we have to show that
$\phi_a(M_{a \Delta_a}(x' \sqcup i')) = f_a(x)$, which can be simplified into
$\neg M_{a \Delta_a}(x' \sqcup i')(r_a) = f_a(x)$. By the definition of an execution
over a module, this can be developed into :

\begin{multline*}
M_{a \Delta_a}(x' \sqcup i')(r_a) = f_{r_a}(M_{a \delta}(x' \sqcup i'), i') \\
= \bigvee\limits_{c \in C_a} \neg M_{a \delta}(x' \sqcup i')(u_c)
= \bigvee\limits_{c \in C_a} \neg c(x_a \mapsto \neg x'(r_a) \sqcup x_b 
\mapsto \neg i'( e_{b, c, a} ) ) \\
= \neg \bigwedge\limits_{c \in C_a} c(x_a \mapsto \neg x'(r_a) \sqcup x_b \mapsto 
\neg i'( e_{b, c, a} ) ).
\end{multline*}

By the above hypothesis, this can be simplified into :

\[M_{a \Delta_a}(x' \sqcup i')(r_a) = \neg \bigwedge\limits_{c \in C_a} c(x_a 
\mapsto x_a \sqcup x_b \mapsto x_b ),\]

which let us simply conclude that :

\[\neg M_{a \Delta_a}(x' \sqcup i')(r_a) = \bigwedge\limits_{c \in C_a} c(x)
= f_a(x)\]

wich proves the lemma. From this result and the fact that the property that
function are locally defined by disjunctive functions isn't broken by
any wiring, we conclude the result.\qed

\end{proof}

\begin{corollary}
Let $F$ be a BAN. There exists $F'$ such that $F' \sim F$ and every function of
$F'$ is monotone.
\end{corollary}

\begin{proof}

To prepare this proof we must first obtain the following result.

\begin{lemma}
\label{monotony-lemma}
Let $x : S \to \B$. Let $f$ be a Boolean function over $S$.
Let $S' = \{s, s^- \mid s \in S\}$. There exists $f'$
a monotone Boolean function over $S'$ such that
$f(x) = f'( x \sqcup s^- \mapsto \neg x(s))$.
\end{lemma}

For reminder, we assume that $x \leq x'$ if and only if $x(s) \leq x'(s)$ for every
$s \in S$, and that $f'$ is monotone if and only if $x \leq x' \implies
f'(x) \leq f'(x')$.

For $x'$ an execution over $S'$, and $s \in S$, we note
$code(x', s) \Leftrightarrow x'(s) = \neg x'(s^-)$. 
Let $f$ be a Boolean function over $S$.

We define $f'$ over the set $S'$
as the following :
\begin{equation*}
f'(x') =\begin{array}\{{ll}.
	f(x'\ress{S}) & \text{ if for every } s \in S, code(x',s)\\
	1 & \text{ if for every } s \in S, \neg code(x', s) \implies x'(s) = x'(s^-) = 1\\
	0 & \text{ otherwise}
	\end{array}.
\end{equation*}

From this definition we clearly see that for all configurations $x$ over $S$,
$f(x) = f'( x \sqcup s^- \mapsto \neg x(s))$. Let us now show that $f'$ is
monotone.

Let $x'$ and $x"$ be two configurations over $S'$, such that $x' < x"$. This
implies that for all $s' \in S'$, $x'(s') \leq x"(s')$ and that there is at
least one $s' \in S'$ such that $x'(s') < x"(s')$. This clearly implies that
the propositions $\forall s \in S, code(x',s)$ and $\forall s \in S, code(x",s)$
cannot both be true.

Let us suppose $\forall s \in S, code(x',s)$ and
$\exists s \in S, \neg code(x",s)$. As $x' < x"$, for every $s \in S$ such that
$\neg code(x", s)$, we now that $x"(s) = x"(s^-) = 1$. This implies that
$f'(x") = 1$, and that $f'(x') \leq f'(x")$.

Let us now suppose that $\exists s \in S, \neg code(x',s)$ and
$\forall s \in S, code(x",s)$. By a similar argument, we now suppose that for 
every $s \in S$ such that $\neg code(x', s)$, we have that
$x'(s) = x'(s^-) = 0$. This implies that $f'(x') = 0$, and $f'(x') \leq f'(x")$.

Let us finally suppose that $\exists s \in S, \neg code(x',s)$ and
$\exists s \in S, \neg code(x",s)$. In this case, we know that $f'(x') = 1
\implies f'(x") = 1$ since $x' < x"$. Assuming $f'(x') = 0$ naturally implies
$f'(x') \leq f'(x")$. This concludes the proof of Lemma \ref{monotony-lemma}.

Let $F$ be a BAN defined over set $S$. For every $a \in S$, we define
$M_a = (T_a, E_a, \alpha_a)$ a module with $T_a = \{u_{a,-}, u_{a,+}\}$,
$E_a = \{e_{b, a, +}, e_{b, a, -} \mid x_b \text{ is included in } f_a\}$. The
function $\alpha$ is such that
$e_{b, a, +} \in E_a \implies e_{b, a, +} \in \alpha(u_{a,+})$ and
$e_{b, a, -} \in E_a \implies e_{b, a, -} \in \alpha(u_{a,-})$.

Let $S$ be a configuration over $S$. We define the monotone function $f'_a$
over the set $\{s, s^- \mid s \in S\}$ that for every configuration $x$ verifies
$f_a(x) = f'_a( x \sqcup s^- \mapsto \neg x(s))$. The existence of such a function
is given by Lemma \ref{monotony-lemma}.

For $x'$ a configuration over $T_a$, and $i$ a configuration over $E_a$,
We define $M_a(u_{a,+})$ as a function that verifies :

\begin{multline*}
M_a(u_{a,+})(x' \sqcup i) =\\
f'_a (
a \mapsto x'(u_{a,+}) \sqcup
a^- \mapsto x'(u_{a,-}) \sqcup
\bigsqcup\limits_{b \neq a} 
\left(b \mapsto i'(e_{b,a,+}) \sqcup b^- \mapsto i'(e_{b,a,-}) \right))
\end{multline*}

The function $M_a(u_{a,-})$ is given by $M_a(u_{a,-})(x' \sqcup i)
= \neg M_a(u_{a,+})(x' \sqcup i)$.

This local module is composed of two automata, one that computes the original
function and one that computes the negation of the original function. This
allows us to simulate the original network while being locally monotone. The
monotony is given by the fact that the configurations used for simulation
are now incomparable to each other. A representation of an example is presented
in Figure~\ref{fig:local-monotone-module}.

\begin{figure}[t!] 
	\begin{center}
		\begin{tikzpicture}[->,>=stealth',shorten
>=1pt,auto,node distance=2cm, semithick, initial text=,inner sep=0pt, minimum
size=0pt] 
		\node[state] (A+) {$u_{a, +}$};
		\node[state] (A-) [right of=A+] {$u_{a, -}$};

		\path (A+) edge[bend left] (A-)
		      (A+) edge[loop above] (A+);
		\path (A-) edge[bend left] (A+)
		      (A-) edge[loop above] (A-);

		\draw (A+) -- (0, -1); 
		\draw (A-) -- (2, -1); 
		\node at (0, -1.2) {$f_a(x)$};
		\node at (2, -1.2) {$\neg f_a(x)$};

		\draw (-1, 0.4) -- (A+);
		\draw (-1, -0.4) -- (A+);
		\draw (3, 0.4) -- (A-);
		\draw (3, -0.4) -- (A-);
		\node at (-1.4, 0.4) {$e_{b, a, +}$};
		\node at (-1.4, -0.4) {$e_{c, a, +}$};
		\node at (3.4, 0.4) {$e_{b, a, -}$};
		\node at (3.4, -0.4) {$e_{c, a, -}$};

		\end{tikzpicture}
	\end{center} 
	\caption{Interaction graph of the locally monotone module for the
example function $f_a(x) = x_a \wedge ( \neg x_b \vee x_c )$. As $x_a$ is
present in the local function, the two automaton composing this module
loop between each other and themselves.}
	\label{fig:local-monotone-module}
\end{figure}
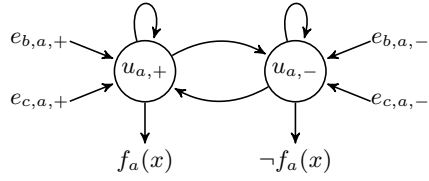

We define $U_{b, a}$ such that $U_{b, a} = T_a$ if the variable $x_b$ is
included in function $f_a$, and $U_{b, a} = \emptyset$ otherwise.

The encodings $\phi_a$ and $\phi_{b,a}$ for every $b$ such that
$U_{b, a} \neq \emptyset$ are defined by :
\begin{equation*}
\phi_a(x') = \phi_{b, a}(x') =\begin{array}\{{ll}.
	1 & \text{if } x'(u_{a, +}) = 1 \text{ and } x'(u_{a, -}) = 0\\
	0 & \text{if } x'(u_{a, +}) = 0 \text{ and } x'(u_{a, -}) = 1\\
	\bullet & \text{otherwise}
	\end{array}.
\end{equation*}

For every $b$ such that $U_{b, a} \neq \emptyset$, we define
$I_{b,a}(e_{b, a, +}) = u_{b, +}$ and $I_{b,a}(e_{b, a, -}) = u_{b, -}$. In
other words, the positive (resp. negative) value of automaton $b$ is given by
the value of the positive (resp. negative) node of the local module $M_b$.

\begin{lemma}
\label{monotony-simulation-lemma}
$M_a$ locally simulates $f_a$ in a input-first way.
\end{lemma}

Let $\Delta_a = \{T_a\}$ be an input-first way update mode for the module
$M_a$. Let $x$ be a configuration over $F$. Let $x'$ be a configuration over 
$T_a$ such that $\phi(x') = x_a$. Let $i'$ be an input configuration over $E_a$
such that for any $b \neq a$, $\phi_{b, a}(i' \circ I^{-1}_{b, a}) = x_b$.

Such a $x'$ verifies $x'(u_{a,+}) = x_a$ and $x'(u_{a, -}) = \neg x_a$.
Such a $i'$ verifies $i'(e_{b, a, +}) = x_b$ and $i'(e_{b, a, -}) = \neg x_b$ for
every $b \neq a$. Theses configurations are well defined.

To prove Lemma \ref{monotony-simulation-lemma}, we have to show that
$\phi_a(M_{a \Delta_a}(x' \sqcup i')) = f_a(x)$. This is equivalent to :

\begin{equation*}
\Leftrightarrow \begin{array}\{{ll}.
	M_a(u_{+, a})(x' \sqcup i') = f_a(x)\\
	M_a(u_{-, a})(x' \sqcup i') = \neg f_a(x)
\end{array}
\end{equation*}

\begin{equation*}
\Leftrightarrow \begin{array}\{{ll}.
	M_a(u_{+, a})(x' \sqcup i') = f_a(x)\\
	\neg M_a(u_{+, a})(x' \sqcup i') = \neg f_a(x)
\end{array}
\end{equation*}

\begin{multline*}
\Leftrightarrow M_a(u_{+, a})(x' \sqcup i') = f_a(x)\\
\Leftrightarrow f'_a (
a \mapsto x'(u_{a,+}) \sqcup
a^- \mapsto x'(u_{a,-}) \sqcup
\bigsqcup\limits_{b \neq a} 
\left(b \mapsto i'(e_{b,a,+}) \sqcup b^- \mapsto i'(e_{b,a,-}) \right))\\
= f_a(x).
\end{multline*}

We noticed earlier that $x'(u_{a,+}) = \neg x'(u_{a,-})$ and that
$i'(e_{b,a,+}) = \neg i'(e_{b,a,-})$ for every $a \neq b$. This implies that
our this evaluation of $f'_a$ can be developed as follows :

\begin{multline*}
f'_a (
a \mapsto x'(u_{a,+}) \sqcup
a^- \mapsto x'(u_{a,-}) \sqcup
\bigsqcup\limits_{b \neq a} 
\left(b \mapsto i'(e_{b,a,+}) \sqcup b^- \mapsto i'(e_{b,a,-}) \right))\\
= f_a(a \mapsto x'(u_{a,+}) \sqcup
\bigsqcup\limits_{b\neq a} b \mapsto i'(e_{b, a,+})) = f_a(a \mapsto x_a \sqcup 
\bigsqcup\limits_{b \neq a} b \mapsto x_b)\\
= f_a(x),
\end{multline*}

which concludes the proof of the Lemma \ref{monotony-simulation-lemma}. Using
this lemma, knowing that the Lemma \ref{monotony-lemma} implies the monotony
of each function in the local modules and the simple fact that local monotony
is not broken by any wiring, we use Theorem \ref{th-simulation} to conclude
this proof.\qed

\end{proof}

\end{document}